\documentclass{article}
\usepackage{amsmath,amsthm,amssymb}
\newenvironment{acknowledgements}{\paragraph{Acknowledgement}}{}
\numberwithin{equation}{section}

\allowdisplaybreaks[3]

\newtheorem{theorem}{Theorem}[section]

\newtheorem{lemma}[theorem]{Lemma}
\newtheorem{proposition}[theorem]{Proposition}
\theoremstyle{definition}
\newtheorem{definition}[theorem]{Definition}
\theoremstyle{remark}
\newtheorem{remark}[theorem]{Remark}

\newcommand\R{{\mathbb R}}
\newcommand\X{{\R^d}}
\newcommand\N{{\mathbb N}}

\newcommand\B{{\mathcal B}}
\newcommand\Bc{\B_{\mathrm{b}}}
\newcommand\Bbs{B_{\mathrm{bs}}}

\renewcommand\L{{\mathcal L}}
\newcommand\K{{\mathcal K}}

\renewcommand\a{{\alpha}}
\newcommand\aC{{\a C}}
\newcommand\La{\Lambda}
\newcommand\la{\lambda}
\newcommand\Ga{\Gamma}
\newcommand\ga{\gamma}
\newcommand\eps{\varepsilon}
\newcommand{\1}{1\!\!1}
\newcommand\n{{|\eta|}}

\newcommand\lv{\left\vert}
\newcommand\rv{\right\vert}
\newcommand\lV{\left\Vert}
\newcommand\rV{\right\Vert}
\newcommand\lu{\left\langle}
\newcommand\ru{\right\rangle}

\newcommand\lluu{\lu\!\!\lu}
\newcommand\rruu{\ru\!\!\ru}
\newcommand\KK{\overline{\K_\aC}}
\newcommand\Ka{\KK}
\newcommand\hT{{\hat{T}}}

\newcommand\hP{{\hat{P}}}
\newcommand\hL{{\hat{L}}}

\newcommand\goto{\rightarrow}

\begin{document}

\title{Correlation~functions evolution for~the~Glauber~dynamics in~continuum}

\author{Dmitri Finkelshtein\thanks{%
Institute of Mathematics, National Academy of Sciences of Ukraine, Kyiv,
Ukraine (\texttt{fdl@imath.kiev.ua}).} \and Yuri Kondratiev\thanks{%
Fakult\"{a}t f\"{u}r Mathematik, Universit\"{a}t Bielefeld, 33615 Bielefeld,
Germany (\texttt{kondrat@math.uni-bielefeld.de})} \and Oleksandr Kutoviy%
\thanks{%
Fakult\"{a}t f\"{u}r Mathematik, Universit\"{a}t Bielefeld, 33615
Bielefeld, Germany (\texttt{kutoviy@math.uni-bielefeld.de}).}}

\maketitle

\begin{abstract}
We construct a correlation functions evolution corresponding to the Glauber
dynamics in continuum. Existence of the corresponding strongly
continuous contraction semigroup in a proper Banach space is shown.
Additionally we prove the existence of the evolution of states and study
their ergodic  properties.
\end{abstract}

\section{Introduction}

Among all birth-and-death Markov processes on configuration spaces
in continuum the Glauber type stochastic dynamics is the object of
permanent  interest for discovering. These dynamics have the given
reversible states which are grand canonical Gibbs measures. This
gives a standard way to construct properly associated stationary
Markov processes using the corresponding (non-local) Dirichlet forms
related to the considered Markov generators and Gibbs measures.
These processes describe the equilibrium Glauber dynamics which
preserve the initial Gibbs state in the time evolution, see, e.g.,
\cite{KoLy}, \cite{KoLyRo}, \cite{KoMiZh}, \cite{FiKoLy}. Note that,
in applications, the time evolution of initial state is the subject
of the primary interest. Therefore, we understand the considered
stochastic (non-equilibrium) dynamics as the evolution of initial
distributions for the system. Actually, the Markov process (provided
it exists) itself gives a general technical equipment to study this
problem. However, we note that the transition from the micro-state
evolution corresponding to the given initial configuration  to the
macro-state dynamics is the well developed concept in the theory of
infinite particle systems. This point of view appeared initially in
the framework of the Hamiltonian dynamics of classical gases, see,
e.g., \cite{DoSiSu}.

The study of the non-equilibrium Glauber dynamics needs construction
of the time evolution for a wider class of initial measures. The
lack of the general Markov processes techniques for the considered
systems makes necessary to develop alternative approaches to study
the state evolutions in the Glauber dynamics. The approach realized
in \cite{KoKutMi}, \cite{KoKutZh}, \cite{FKKZh} is probably the only known at the
present time. The description of the time evolutions for measures on
configuration spaces in terms of an infinite system of evolutional
equations for the corresponding correlation functions was used
there. The latter system is a Glauber evolution's analog of the
famous BBGKY-hierarchy for the Hamiltonian dynamics.

Here we extend constructive approach developed in \cite{FKKZh} to correlation function
evolution of the Glauber dynamics in
continuum. We describe a reasonable Banach space where the evolution problem
can be solved. Moreover, we construct an explicit approximation  by bounded operators of the corresponding
evolutional semigroup. We prove
that functions in this evolution stay correlation functions of some measures
(states) on configuration spaces; this means that we show the
existence of states evolution. At the end we obtain the ergodic properties
of the state evolution.

\section{Basic facts and notation}

Let ${\mathcal{B}}({{\mathbb{R}}^d})$ be the family of all Borel sets in ${{%
\mathbb{R}}^d}$, $d\geq 1$; ${\mathcal{B}}_{\mathrm{b}}
({{\mathbb{R}}^d})$ denote the system of all bounded sets in
${\mathcal{B}}({{\mathbb{R}}^d})$.

The configuration space over space ${{\mathbb{R}}^d}$ consists of all
locally finite subsets (configurations) of ${{\mathbb{R}}^d}$, namely,
\begin{equation}  \label{confspace}
\Gamma =\Gamma_{{\mathbb{R}}^d} :=\Bigl\{ \gamma \subset {{\mathbb{R}}^d}
\Bigm| |\gamma \cap \Lambda |<\infty, \text{ for all } \Lambda \in {\mathcal{%
B}}_{\mathrm{b}} ({{\mathbb{R}}^d})\Bigr\}.
\end{equation}
The space $\Gamma$ is equipped with the vague topology, i.e., the minimal
topology for which all mappings $\Gamma\ni\gamma\mapsto \sum_{x\in\gamma}
f(x)\in{\mathbb{R}}$ are continuous for any continuous function $f$ on ${{%
\mathbb{R}}^d}$ with compact support; note that the summation in
$\sum_{x\in\gamma} f(x)$ is taken over only
finitely many points of $\gamma$ which belong to the support of $f$. In \cite%
{KoKut}, it was shown that $\Gamma$ with the vague topology may be
metrizable and it becomes a Polish space (i.e., complete separable
metric space). Corresponding to this topology, Borel $\sigma
$-algebra ${\mathcal{B}}(\Gamma )$ is the smallest $\sigma $-algebra
for which all mappings $\Gamma \ni \gamma \mapsto |\gamma_ \Lambda |\in{%
\mathbb{N}}_0:={\mathbb{N}}\cup\{0\}$ are measurable for any $\Lambda\in{%
\mathcal{B}}_{\mathrm{b}}({{\mathbb{R}}^d})$. Here $\gamma_\Lambda:=\gamma%
\cap\Lambda$, and $|\cdot|$ means the cardinality of a finite set.

The space of $n$-point configurations in an arbitrary
$Y\in{\mathcal{B}}({{\mathbb{R}}^d})$ is defined by
\begin{equation*}
\Gamma^{(n)}_Y:=\Bigl\{  \eta \subset Y \Bigm| |\eta |=n\Bigr\} ,\quad n\in {%
\mathbb{N}}.
\end{equation*}
We set also $\Gamma^{(0)}_Y:=\{\emptyset\}$. As a set,
$\Gamma^{(n)}_Y$ may be identify with the symmetrization of
\begin{equation*}
\widetilde{Y^n} = \Bigl\{ (x_1,\ldots ,x_n)\in Y^n \Bigm| x_k\neq x_l \text{
if } k\neq l\Bigr\} .
\end{equation*}

Hence one can introduce the corresponding Borel $\sigma $-algebra,
which we denote by ${\mathcal{B}}(\Gamma^{(n)}_Y)$. The space of
finite configurations in an arbitrary
$Y\in{\mathcal{B}}({{\mathbb{R}}^d})$ is defined by
\begin{equation*}
\Gamma_{0,Y}:=\bigsqcup_{n\in {\mathbb{N}}_0}\Gamma^{(n)}_Y.
\end{equation*}
This space is equipped with the topology of disjoint unions.
Therefore, one can introduce the corresponding Borel $\sigma
$-algebra ${\mathcal{B}} (\Gamma _{0,Y})$. In the case of
$Y={{\mathbb{R}}^d}$ we will omit the index $Y$ in
the notation, namely, $\Gamma_0:=\Gamma_{0,{{\mathbb{R}}^d}}$, $%
\Gamma^{(n)}:=\Gamma^{(n)}_{{{\mathbb{R}}^d}}$.

The restriction of the Lebesgue product measure $(dx)^n$ to $\bigl(%
\Gamma^{(n)}, {\mathcal{B}}(\Gamma^{(n)})\bigr)$ we denote by $m^{(n)}$. We
set $m^{(0)}:=\delta_{\{\emptyset\}}$. The Lebesgue--Poisson measure $%
\lambda $ on $\Gamma_0$ is defined by
\begin{equation}  \label{LP-meas-def}
\lambda :=\sum_{n=0}^\infty \frac {1}{n!}m^{(n)}.
\end{equation}
For any $\Lambda\in{\mathcal{B}}_{\mathrm{b}}({{\mathbb{R}}^d})$ the
restriction of $\lambda$ to $\Gamma _\Lambda:=\Gamma_{0,\Lambda}$ will be
also denoted by $\lambda $. The space $\bigl(
\Gamma, {\mathcal{B}}(\Gamma)\bigr)$ is the projective limits of the family
of spaces $\bigl\{( \Gamma_\Lambda, {\mathcal{B}}(\Gamma_\Lambda))\bigr\}%
_{\Lambda \in {\mathcal{B}}_{\mathrm{b}} ({{\mathbb{R}}^d})}$. The Poisson
measure $\pi$ on $\bigl(\Gamma ,{\mathcal{B}}(\Gamma )\bigr)$ is given as
the projective limit of the family of measures $\{\pi^\Lambda \}_{\Lambda
\in {\mathcal{B}}_{\mathrm{b}} ({{\mathbb{R}}^d})}$ where $%
\pi^\Lambda:=e^{-m(\Lambda)}\lambda $ is the probability measure on $\bigl( %
\Gamma_\Lambda, {\mathcal{B}}(\Gamma_\Lambda)\bigr)$. Here $m(\Lambda)$ is
the Lebesgue measure of $\Lambda\in {\mathcal{B}}_{\mathrm{b}} ({{\mathbb{R}}%
^d})$.

For any measurable function $f:{{\mathbb{R}}^d}\rightarrow {\mathbb{R}}$ we
define a \emph{Lebesgue--Poisson exponent}
\begin{equation}  \label{LP-exp-def}
e_\lambda(f,\eta):=\prod_{x\in\eta} f(x),\quad \eta\in\Gamma_0; \qquad
e_\lambda(f,\emptyset):=1.
\end{equation}
Then, by \eqref{LP-meas-def}, for $f\in L^1({{\mathbb{R}}^d},dx)$ we
obtain $e_\lambda(f)\in L^1(\Gamma_0,d\lambda)$ and
\begin{equation}  \label{LP-exp-mean}
\int_{\Gamma_0}
e_\lambda(f,\eta)d\lambda(\eta)=\exp\biggl\{\int_{{\mathbb{R}}^d}
f(x)dx\biggr\}.
\end{equation}

A set $M\in {\mathcal{B}} (\Gamma_0)$ is called bounded if there exists $%
\Lambda \in {\mathcal{B}}_{\mathrm{b}} ({{\mathbb{R}}^d})$ and $N\in {%
\mathbb{N}}$ such that $M\subset \bigsqcup_{n=0}^N\Gamma _\Lambda^{(n)}$.
We will use the following classes of functions on $\Ga_{0}$:
$L_{\mathrm{ls}}^0(\Ga _0)$ is the set of all measurable functions
on $\Ga_0$ which have a local support, i.e. $G\in
L_{\mathrm{ls}}^0(\Ga _0)$ if there exists $\La \in \B_b({\R}^{d})$
such that $G\upharpoonright_{\Ga _0\setminus \Ga _\La }=0$;
$B_{\mathrm{bs}}(\Ga _0)$ is the set of bounded measurable functions
with bounded support, i.e. $G\upharpoonright_{\Ga _0\setminus B}=0$
for some bounded $B\in \B(\Ga_0)$.

Any ${\mathcal{B}}(\Gamma_0)$-measurable function $G$ on $%
\Gamma_0$, in fact, is a sequence of functions $\bigl\{G^{(n)}\bigr\}_{n\in{%
\mathbb{N}}_0}$ where $G^{(n)}$ is a ${\mathcal{B}}(\Gamma^{(n)})$%
-measurable function on $\Gamma^{(n)}$.

On $\Ga $ we consider the set of cylinder functions
$\mathcal{F}_{\mathrm{cyl}}(\Ga )$. These functions are
characterized by the following relation: $F(\ga )=F\upharpoonright
_{\Ga _\La }(\ga _\La )$.

There is the following mapping from $L_{\mathrm{ls}}^0(\Ga _0)$ into ${{%
\mathcal{F}}_{\mathrm{cyl}}}(\Gamma )$, which plays the key role in
our further considerations:
\begin{equation}
KG(\gamma ):=\sum_{\eta \Subset \gamma }G(\eta ), \quad \gamma \in \Gamma,
\label{KT3.15}
\end{equation}
where $G\in L_{\mathrm{ls}}^0(\Ga _0)$, see, e.g.,
\cite{KoKu99,Le75I,Le75II}. The summation in \eqref{KT3.15} is taken
over all finite subconfigurations $\eta\in\Ga_0$ of the (infinite)
configuration $\gamma\in\Ga$; we denote this, by the symbol,
$\eta\Subset\gamma $. The mapping $K$ is linear, positivity
preserving, and invertible, with
\begin{equation}
K^{-1}F(\eta ):=\sum_{\xi \subset \eta }(-1)^{|\eta \setminus \xi |}F(\xi
),\quad \eta \in \Gamma_0.  \label{k-1trans}
\end{equation}
Here and in the sequel inclusions like $\xi\subset\eta$ hold for $%
\xi=\emptyset$ as well as for $\xi=\eta$. We denote the restriction
of $K$ onto functions on $\Gamma_0$ by $K_0$.

For any fixed $C>1$ we consider the following space of
${\mathcal{B}} (\Gamma_0)$-measurable functions
\begin{equation}  \label{norm}
\L _C:=\biggl\{ G:\Gamma_0\rightarrow{\mathbb{R}} \biggm| \|G\|_C:=
\int_{\Gamma_0} |G(\eta)| C^{|\eta|} d\lambda(\eta) <\infty\biggr\}.
\end{equation}
In the sequel, $\mathcal{L}_{C}^{\mathrm{ls}}$ denotes the set $L_{\mathrm{ls}}^0(\Ga _0)\cap\mathcal{L}_{C}$.
The space $\mathcal{L}_{C}$ can be made into a Banach space in a standard way; simply taking the quotient space with respect to the kernel of $\|\cdot\|_C$. To simplify notations, we use the same symbol $\L _C$ for the corresponding Banach space.

A measure $\mu \in {\mathcal{M}}_{\mathrm{fm} }^1(\Gamma )$ is called
locally absolutely continuous with respect to (w.r.t. for short) Poisson measure $\pi$ if for any $%
\Lambda \in {\mathcal{B}}_{\mathrm{b}} ({{\mathbb{R}}^d})$ the projection of
$\mu$ onto $\Gamma_\Lambda$ is absolutely continuous w.r.t. projection of $%
\pi$ onto $\Gamma_\Lambda$. By \cite{KoKu99}, in this case, there
exists a \emph{correlation functional} $k_{\mu}:\Gamma_0 \rightarrow
{\mathbb{R}}_+$ such that for any $G\in B_{\mathrm{bs}} (\Gamma_0)$
the following equality holds
\begin{equation}  \label{eqmeans}
\int_\Gamma (KG)(\gamma) d\mu(\gamma)=\int_{\Gamma_0}G(\eta)
k_\mu(\eta)d\lambda(\eta).
\end{equation}
Restrictions $k_\mu^{(n)}$ of this functional on $\Gamma_0^{(n)}$, $n\in{%
\mathbb{N}}_0$ are called \emph{correlation functions} of the
measure $\mu$. Note that $k_\mu^{(0)}=1$.

We recall now without a proof the partial case of the well-known
technical lemma (cf., \cite{KoMiZh}) which plays very important role
in our calculations.

\begin{lemma}
\label{Minlos} For any measurable function $H:\Gamma_0\times\Gamma_0\times%
\Gamma_0\rightarrow{\mathbb{R}}$
\begin{equation}  \label{minlosid}
\int_{\Gamma _{0}}\sum_{\xi \subset \eta }H\left( \xi ,\eta \setminus \xi
,\eta \right) d\lambda \left( \eta \right) =\int_{\Gamma _{0}}\int_{\Gamma
_{0}}H\left( \xi ,\eta ,\eta \cup \xi \right) d\lambda \left( \xi \right)
d\lambda \left( \eta \right)
\end{equation}
if only both sides of the equality make sense.
\end{lemma}

\section{Non-equilibrium Glauber dynamics in continuum}

\label{dualconstraction}

Let $\phi:{{\mathbb{R}}^d}\rightarrow{\mathbb{R}}_+:=[0;+\infty)$ be
an even non-negative function which satisfies the following
integrability condition
\begin{equation}  \label{weak_integrability}
C_\phi := \int_{{\mathbb{R}}^d} \bigl(1-e^{-\phi(x)}\bigr) dx < +\infty
\end{equation}
For any $\gamma\in\Gamma$, $x\in{{\mathbb{R}}^d}\setminus\gamma$ we set
\begin{equation}  \label{relativeenergy}
E^\phi(x,\gamma) :=\sum_{y\in\gamma} \phi(x-y) \in [0;\infty].
\end{equation}

Let us define the (pre-)generator of the Glauber dynamics: for any $F\in{{%
\mathcal{F}}_{\mathrm{cyl}}}(\Gamma )$ we set
\begin{align}
(LF)(\gamma):=&\sum_{x\in\gamma} \bigl[F(\gamma\setminus x) -F(\gamma)\bigr]
\label{genGa} \\
&+ z \int_{{{\mathbb{R}}^d}} \bigl[F(\gamma\cup x) -F(\gamma)\bigr]\exp%
\bigl\{-E^\phi(x,\gamma)\bigr\} dx, \qquad \gamma\in\Gamma.  \notag
\end{align}
Here $z>0$ is the \textit{activity} parameter. Note that for any
$F\in{{\mathcal{ F}}_{\mathrm{cyl}}}(\Gamma)$ there exists
$\La\in\Bc(\X)$ such that $F(\gamma\setminus x)=F(\gamma)$
for any $x\in\gamma_{\Lambda^c}$ and $F(\gamma\cup x)=F(\gamma)$ for any $%
x\in\Lambda^c$; note also that $\exp\bigl\{-E^\phi(x,\gamma)\bigr\}\leq 1$; therefore, sum and integral in \eqref{genGa} are finite.

In \cite{FKKZh}, it was shown that the mapping ${\hat{L}}
:=K^{-1}LK$ given on $B_{\mathrm{bs}}(\Gamma_0)$ by the following
expression
\begin{align}  \label{Lhat}
({\hat{L}} G)(\eta) =&- |\eta| G(\eta) \\&+ z \sum_{\xi\subset\eta}\int_{{%
\mathbb{R}}^d} e^{-E^\phi(x,\xi)} G(\xi\cup x)e_\lambda (e^{-\phi (x -
\cdot)}-1,\eta\setminus\xi) dx \notag
\end{align}
is a linear operator on $\L _C$ with the dense domain $D({\hat{L}} )=\L %
_{2C}\subset\L _C$. If, additionally,
\begin{equation}  \label{verysmallparam}
z\leq \min\bigl\{Ce^{-CC_{\phi }} ; \, 2Ce^{-2CC_{\phi }}\bigr\},
\end{equation}
then $\bigl({\hat{L}} , D({\hat{L}} )\bigr)$ is closable linear operator in $%
\L _C$ and its closure (which we denote by $\hat{L}$ also) generates
a strongly continuous contraction semigroup ${\hat{T}} (t)$ on
$\L_C$.

Let us define $d\lambda_C:= C^{|\cdot|} d\lambda$. The
topologically dual space is the space $(\L _C)^{\prime
}=\bigl(L^1(\Gamma_0, d\lambda_C)\bigr)^{\prime} =L^\infty(\Gamma_0,
d\lambda_C)$. The space $L^\infty(\Gamma_0, d\lambda_C)$ is
isometrically isomorphic to the Banach space
\begin{equation*}
{\mathcal{K}}_{C}:=\left\{k:\Gamma_{0}\rightarrow{\mathbb{R}}\,\Bigm| k\cdot
C^{-|\cdot|}\in L^{\infty}(\Gamma_{0},\lambda)\right\}
\end{equation*}
with the norm
$
\|k\|_{{\mathcal{K}}_C}:=\bigl\|C^{-|\cdot|}k(\cdot)\bigr\|_{L^{\infty}(\Gamma_{0},%
\lambda)},
$
where the isomorphism is provided by the isometry $R_C$
\begin{equation}  \label{isometry}
(\L _C)^{\prime }\ni k \longmapsto R_Ck:=k\cdot C^{|\cdot|}\in {\mathcal{K}}%
_C.
\end{equation}

In fact, we may say about a duality between Banach spaces $\L _C$
and ${\mathcal{K}}_C$, which is given by the following expression
\begin{equation}
\left\langle\!\left\langle G,\,k \right\rangle\!\right\rangle :=
\int_{\Gamma_{0}}G\cdot k\, d\lambda,\quad G\in\L _C, \ k\in {\mathcal{K}}_C
\label{duality}
\end{equation}
with
\begin{equation}
\left\vert \left\langle\!\left\langle G,k \right\rangle\!\right\rangle
\right\vert \leq \|G\|_C \cdot\|k\|_{{\mathcal{K}}_C}.  \label{funct_est}
\end{equation}
It is clear that for any $k\in {\mathcal{K}}_C$
\begin{equation}  \label{RB-norm}
|k(\eta)|\leq \|k\|_{{\mathcal{K}}_C} \, C^{|\eta|} \quad \text{for } \lambda%
\text{-a.a. } \eta\in\Gamma_0.
\end{equation}

Let $\bigl( {\hat{L}} ^{\prime }, D({\hat{L}} ^{\prime })\bigr)$ be
an operator in $(\L_C)^{\prime }$ which is dual to the closed
operator $\bigl( {\hat{L}} , D({\hat{L}} )\bigr)$. We consider also
its image in ${\mathcal{K}}_C$ under isometry $R_C$, namely, let ${\hat{L}}^{*}=R_C{%
\hat{L}} ^{\prime }R_{C^{-1}}$ with a domain $D({\hat{L}} ^{*})=R_C  D({\hat{L}%
} ^{\prime })$. Then, for any $G\in D(\hat{L})$, $k\in D({\hat{L}}^\ast)$
\begin{align*}
\int_{\Gamma_0}G\cdot {\hat{L}}^\ast k d\lambda&=\int_{\Gamma_0}G\cdot R_C{%
\hat{L}} ^{\prime }R_{C^{-1}} k d\lambda=\int_{\Gamma_0}G\cdot {\hat{L}}
^{\prime }R_{C^{-1}} k d\lambda_C \\
&= \int_{\Gamma_0}{\hat{L}} G\cdot R_{C^{-1}} k d\lambda_C=\int_{\Gamma_0}{%
\hat{L}} G\cdot k d\lambda,
\end{align*}
therefore, ${\hat{L}}^\ast$ is the dual operator to ${\hat{L}} $ w.r.t.
the duality \eqref{duality}.

By \cite{FiKoOl07}, we have the precise form of ${\hat{L}}^{*}$ on $D({\hat{L}}^\ast)$:
\begin{align}  \label{dual-descent}
({\hat{L}}^* k)(\eta)=&-\vert \eta \vert k(\eta) \\
&+z\sum_{x\in \eta}e^{-E^\phi (x,\eta\setminus x)} \int_{\Gamma_0}e_\lambda
(e^{-\phi (x - \cdot)}-1,\xi) k((\eta\setminus x)\cup\xi)\,d\lambda (\xi).
\notag
\end{align}

Under condition \eqref{verysmallparam}, we consider the adjoint semigroup ${%
\hat{T}} ^{\prime }(t)$ in $(\L _C)^{\prime }$ and its image ${\hat{T}}%
^\ast(t)$ in ${\mathcal{K}}_C$. Now, we may apply general results
about adjoint semigroups (see, e.g., \cite{EN}) onto the semigroup
${\hat{T}}^\ast(t)$. The last semigroup will be weak*-continuous,
moreover, weak*-differentiable at $0$ and ${\hat{L}}^\ast$ will be
weak*-generator of ${\hat{T}}^\ast(t)$. Here and below we mean
``weak*-properties'' w.r.t. duality \eqref{duality}. Let
\begin{equation}
\mathring{\K}_C=\Bigl\{ k\in{\mathcal{K}}_C \Bigm| \exists
\lim_{t\downarrow0}\| {\hat{T}}^\ast(t)k - k\|_{{\mathcal{K}}_C} =0\Bigr\}.
\end{equation}
Then $\mathring{\K}_C$ is a closed, weak*-dense, ${\hat{T}}^\ast(t)$-invariant
linear subspace of ${\mathcal{K}}_C$. Moreover, $\mathring{\K}_C=\overline{D(%
{\hat{L}}^\ast)}$ (the closure is in the norm of ${\mathcal{K}}_C$). Let ${%
\hat{T}}^\odot(t)$ denote the restriction of ${\hat{T}}^\ast(t)$ onto Banach
space ${\mathring{\K}}_C$. Then ${\hat{T}} ^\odot(t)$ is a $C_0$-semigroup
on ${\mathring{\K}}_C$ and its generator ${\hat{L}} ^\odot$ will be part of $%
{\hat{L}}^\ast$, namely,
\[
D({\hat{L}} ^\odot)=\Bigl\{k\in D({\hat{L}}^\ast) \Bigm| {\hat{L}}^\ast k\in \overline{D({\hat{L}}^\ast)}\Bigr\}
\]
and ${\hat{L}%
}^\ast k ={\hat{L}}^\odot k$ for any $k\in D({\hat{L}}^\odot)$.

And now we consider another ${\hat{T}}^\ast(t)$-invariant subspace. We
present, at first, the useful subspace in $D({\hat{L}}^\ast)$.

\begin{proposition}
\label{domain-adj} For any $\a\in(0;1)$ the following inclusions
hold ${\mathcal{K}}_{\a C}\subset D({\hat{L}}^\ast)\subset
\overline{D({\hat{L}}^\ast)} \subset{\mathcal{K}}_C$.
\end{proposition}

\begin{proof}
Let $\a \in(0;1)$ and $k\in\K_{\a C}$ then, using \eqref{LP-exp-mean}
and \eqref{RB-norm}, for $\la$-a.a. $\eta\in\Ga_0$ we may estimate
\begin{align*}
&\, C^{-\lv  \eta \rv  }\lv  \eta \rv  \lv k\left( \eta \right)
\rv   +\sum_{x\in \eta }\int_{\Ga_{0}}e_{\la }\left( 1-e^{-\phi
\left( x-\cdot \right) },\xi \right) C^{-\lv \eta \rv  }\lv k\left(
\left( \eta \setminus x\right) \cup \xi \right) \rv  d\la
\left( \xi \right)  \\
\leq\,  &C^{-\lv  \eta \rv  }\lv  \eta \rv
\lV  k\rV _{\K_{\a C}}\left( \a C%
\right)^{\lv  \eta \rv  } \\
&+\sum_{x\in \eta }\int_{\Ga _{0}}e_{\la  }\left( 1-e^{-\phi \left(
x-\cdot \right) },\xi \right) C^{-\lv  \eta \rv  }\lV k\rV_{\K_{\a
C}}\left( \a C\right)^{\lv  \left( \eta \setminus x\right) \cup \xi
\rv  }d\la
\left( \xi \right)  \\
=\, &\a^{|\eta|}|\eta| \lV  k\rV _{\K_{\a C}} +\frac{1}{\a
C }\lV  k\rV _{\K_{\aC}}%
\a^{|\eta|}\sum_{x\in \eta }\int_{\Ga _{0}}e_{\la  }\left( \a
C\left( 1-e^{-\phi \left( x-\cdot \right)
}\right) ,\xi \right) d\la  \left( \xi \right)  \\
=\, &\a^{|\eta|}|\eta| \lV  k\rV _{\K_{\a C}}+\frac{1}{\a
C }\lV  k\rV _{\K_{\aC}}%
\a^{|\eta|}\lv  \eta \rv  \exp \left\{ \a C%
C_{\phi }\right\}  \\
\leq\,  &\lV  k\rV _{\K_{\a C}} \frac{-1}{e\ln
\a } \left( 1+\frac{1}{%
\a C }\exp \left\{ \a C C_{\phi }\right\} \right),
\end{align*}
since $x\a^x \leq -\dfrac{1}{e\ln \a }$ for any $\a \in(0;1)$ and
$x\geq0$. Using the definition of $D({\hat{L}}^\ast)$ and Lemma \ref{Minlos} we get immediately  the statement of the proposition.
\end{proof}

\begin{remark}
By the same arguments, the set of all functions $k\in{\mathcal{K}}_C$ such
that
\begin{equation*}
| k(\eta) |\leq \mathrm{const}\cdot \frac{1}{|\eta|}C^{|\eta|}, \quad
\eta\in\Gamma_0\setminus\{\emptyset\}
\end{equation*}
is a subset of $D({\hat{L}}^\ast)$. Due to the elementary inequality
$\a^x <
\mathrm{const} \cdot x^{-1}$ for any $\a \in(0;1)$, $x>0$, we have that this set contains ${%
\mathcal{K}}_{\a C}$. But the smaller set ${\mathcal{K}}_{\a C}$ is
more useful for our calculations.
\end{remark}

\begin{proposition}
\label{invariantspace} Suppose that  (\ref{verysmallparam}) is satisfied.  Furthermore, we additionally assume that
\begin{equation}\label{new_z}
z< C e^{-CC_\phi},\quad\mathrm{if}\quad CC_\phi\leq \ln2.
\end{equation}
Then there exists $\a_0=\a_0(z,\phi,C)\in (0;1)$ such that for any $\a\in (\a_0;1)$ the set
${\mathcal{K}}_{\a C}$ is the ${\hat{T}}^\ast(t)$-invariant linear subspace of
${\mathcal{K}}_C$.
\end{proposition}

\begin{proof}
Let us consider function $f(x):=x e^{-x}$, $x\geq 0$. It has the following properties:  $f$ is increasing on $[0; 1]$ from $0$ to $e^{-1}$ and it is
asymptotically decreasing on $[1;+\infty)$ from $e^{-1}$ to $0$; $f(x) < f(2x)$ for $x \in (0, \ln 2)$; $x=\ln 2$ is the only  non-zero solution to
$f(x)=f(2x)$.

By assumption \eqref{verysmallparam}, $zC_\phi \leq \min\{CC_\phi e^{-CC_\phi},
2CC_\phi e^{-2CC_\phi}\}$. Therefore, if $CC_\phi e^{-CC_\phi}\neq 2CC_\phi e^{-2CC_\phi}$ then \eqref{verysmallparam} with necessity implies
\begin{equation}\label{less_e-1}
z C_\phi < e^{-1}.
\end{equation}
This inequality remains also true if $CC_\phi=\ln 2$ because of \eqref{new_z}.
Under condition \eqref{less_e-1}, the equation $f(x)=z C_\phi$ has
exactly two roots, say, $0<x_1<1<x_2<+\infty$. Then,
\eqref{new_z} implies $x_1< C C_\phi < 2 C C_\phi \leq x_2$.

\noindent If $CC_\phi>1$ then we  set $\a_0:=\max\left\{\frac{1}{2};\frac{1}{C
C_\phi};\frac{1}{C}\right\}<1$. This yields $2\a C C_\phi > C
C_\phi$ and $\a CC_\phi >1>x_{1}$. If $x_{1}<CC_\phi \leq 1$ then we
set $\a_0:=\max\left\{\frac{1}{2};\frac{x_{1}}{C
C_\phi};\frac{1}{C}\right\}<1$ that gives $2\a C C_\phi > C C_\phi$
and $\a CC_\phi >x_{1}$.

\noindent As a result,
\begin{equation}\label{ineq-alpha}
x_1<\a CC_\phi  < CC_\phi <2 \a C C_\phi < 2 C C_\phi \leq x_2
\end{equation}
and $1<\a C<C<2\a C<2C$. The last inequality shows that $\L_{2C
}\subset\L_{2\a C}\subset \L_C\subset \L_{\a C}$. Moreover, by
\eqref{ineq-alpha}, we may prove that the operator $(\hL , \L_{2\a
C})$ is closable in $\L_{\a C}$ and its closure is a generator of a
contraction semigroup $\hT_\a (t)$ on $\L_{\a C}$. The proof is
identical to that in \cite{FKKZh}.

It is easy to see, that $\hT_\a (t) G= \hT (t) G$ for any
$G\in\L_C$. Indeed, from the construction of the semigroup $\hT
(t)$, see \cite{FKKZh}, and analogous construction for the semigroup
$\hT_\a (t)$, we have that there exists family of mappings
$\hP_\delta$, $\delta>0$ independent of $\a $ and $C$, namely,
\begin{align}\label{apprsemigroup}
\bigl(\hat{P}_{\delta } G\bigr) \left( \eta \right) :=&\sum_{\xi
\subset \eta }\left( 1-\delta \right) ^{\left\vert \xi \right\vert
}\int_{\Gamma _{0}}\left( z\delta \right) ^{\left\vert \omega
\right\vert }G\left( \xi \cup \omega \right)
\\&\times \prod\limits_{y\in \xi }e^{-E^\phi\left( y,\omega \right)
}\prod\limits_{y\in \eta \setminus \xi }\left( e^{-E^\phi\left(
y,\omega \right) }-1\right) d\lambda \left( \omega \right) , \quad
\eta\in\Ga_0.\notag
\end{align}
such that $\hP _\delta^{\left[\frac{t}{\delta}\right]}$ for any
$t\geq 0$ strongly converges to $\hT (t)$ and $\hT_\a (t)$ in $\L_C$
and $\L_{\a C}$, correspondingly, as $\delta\goto 0$. Here and below
$[\,\cdot\,]$ means the entire part of a number. Then for any
$G\in\L_{C }\subset\L_{\a C}$ we have that $\hT (t)G\in\L_{C
}\subset\L_{\a C}$ and $\hT_\a (t) G\in\L_{\a C}$ and
\begin{align*}
\| \hT (t)G-\hT_\a (t) G\|_{\a C}&\leq \Bigl\| \hT (t)G-\hP
_\delta^{\left[\frac{t}{\delta}\right]}G\Bigr\|_{\a C} + \Bigl\| \hT_\a (t)
G-\hP_\delta^{\left[\frac{t}{\delta}\right]}G\Bigr\|_{\a C}\\&\leq \Bigl\| \hT
(t)G-\hP_\delta^{\left[\frac{t}{\delta}\right]}G\Bigr\|_{ C} + \Bigl\| \hT_\a
(t) G-\hP_\delta^{\left[\frac{t}{\delta}\right]}G\Bigr\|_{\a C}\goto0,
\end{align*}
as $\delta\goto 0$. Therefore, $\hT (t)G=\hT_\a (t) G$ in $\L_{\a
C}$ (recall that $G\in\L_C$) that yields $ \hT (t)G(\eta)=\hT_\a (t)
G(\eta)$ for $\la$-a.a. $\eta\in\Ga_0$ and, therefore, $\hT
(t)G=\hT_\a (t) G$ in $\L_{C}$.

Note that for any $G\in\L_C\subset\L_{\a C}$ and for any $k\in
\K_{\a C}\subset \K_C$ we have $\hT_\a (t) G\in\L_{\a C}$ and
\begin{equation*}
\lluu   \hT_\a (t) G, k\rruu  =\lluu   G, \hT^\ast_\a (t) k\rruu ,
\end{equation*}
where, by construction, $\hT^\ast_\a (t) k\in\K_{\a C}$. But
$G\in\L_C$, $k\in\K_C$ implies
\begin{equation*}
\lluu   \hT_\a (t) G, k\rruu  =\lluu   \hT (t) G, k\rruu  =\lluu G,
\hT^\ast(t) k\rruu .
\end{equation*}
Hence, $\hT^\ast(t) k = \hT^\ast_\a (t) k\in\K_{\a C}$, $k\in\K_{\a C}$ that proves
the statement.
\end{proof}

\begin{remark}
As a result, \eqref{verysmallparam} implies that for any $k_0\in \overline{D({%
\hat{L}}^\ast)}$ the Cauchy problem in ${\mathcal{K}}_C$
\begin{equation}
\begin{cases}
\dfrac{\partial}{\partial t} k_t = {\hat{L}}^\ast k_t \\\label{cau1}
k_t \bigr|_{t=0} = k_0%
\end{cases}%
\end{equation}
has a unique mild solution: $k_t= {\hat{T}}^\ast (t)k_0= {\hat{T}}^\odot
(t)k_0\in\overline{D({\hat{L}}^\ast)}$. Moreover, $k_0\in{\mathcal{K}}_{\a C}
$ implies $k_t\in{\mathcal{K}}_{\a C}$ provided \eqref{new_z} is satisfied.
\end{remark}

\begin{remark}
The Cauchy problem \eqref{cau1} is well-posed in $\mathring{\K}_C=\overline{D(%
{\hat{L}}^\ast)}$, i.e., for every $k_0\in D({\hat{L}}^\odot)$ there exists a unique solution $k_{t}\in\mathring{\K}_C$ of  \eqref{cau1}.
\end{remark}

Let \eqref{verysmallparam} and \eqref{new_z} be satisfied and let $\alpha_0$ be chosen as in the proof of Proposition~\ref%
{invariantspace} and fixed. Suppose that  $\alpha\in(\alpha_0;1)$. Then,
Propositions~\ref{domain-adj} and \ref{invariantspace} imply
$\overline{{\mathcal{K}}_{\a C}}\subset\overline{D(\hat{L}^{\ast
})}$ and the Banach subspace $\overline{{\mathcal{K}}_{\a C}}$ is
${\hat{T}}^\ast(t)$- and, therefore, ${\hat{T}}^\odot(t)$-invariant
due to the continuity of these operators.

Let now ${\hat{T}}^{\odot\a}(t)$ be the restriction of the strongly
continuous semigroup ${\hat{T}}^\odot(t)$
onto the closed linear
subspace $\overline{{\mathcal{K}}_{\a C}}$.
By general result (see, e.g.,
\cite{EN}), ${\hat{T}} ^{\odot\a}(t)$ is a strongly continuous
semigroups on $\overline{{\mathcal{K}}_{\a C}}$ with generator
${\hat{L}}^{\odot\a}$ which is the restriction of the operator
${\hat{L}}^\odot $. Namely,
\begin{equation}
D({\hat{L}}^{\odot\a})=\Bigl\{k\in \overline{{\mathcal{K}}_{\a C}} \Bigm| {%
\hat{L}}^\ast k\in\overline{{\mathcal{K}}_{\a C}} \Bigr\},
\label{domAdjtimes}
\end{equation}
and
\begin{equation}
{\hat{L}}^{\odot\a} k = {\hat{L}}^\odot k = {\hat{L}}^\ast k, \qquad k\in D({%
\hat{L}}^{\odot\a})  \label{restRRren}
\end{equation}

Since ${\hat{T}}(t)$ is a contraction semigroup on $\L _C$, then, ${\hat{T}}
^{\prime }(t)$ is also a contraction semigroup on $(\L _C)^{\prime }$; but
isomorphism \eqref{isometry} is isometrical, therefore, ${\hat{T}}^\ast(t)$
is a contraction semigroup on ${\mathcal{K}}_C$. As a result, its restriction $%
{\hat{T}}^{\odot\a}(t)$ is a contraction semigroup on $\overline{{\mathcal{K}}_%
{\a C}}$. Note also, that by \eqref{domAdjtimes},
\begin{equation*}
D_{\a C}:=\Bigl\{k\in {\mathcal{K}}_{\a C} \Bigm| {\hat{L}}^\ast k\in%
\overline{{\mathcal{K}}_{\a C}} \Bigr\}
\end{equation*}
is a core for ${\hat{L}}^{\odot\a}$ in $\overline{{\mathcal{K}}_{\a C}}$.

By \eqref{apprsemigroup}, for any $k\in {\mathcal{K}}_{{\a C}}$, $G\in B_{\mathrm{bs}}(\Gamma
_{0})$ we have
\begin{align*}
& \int_{\Gamma _{0}}(\hat{P}_{\delta }G)\left( \eta \right) k\left( \eta
\right) d\lambda \left( \eta \right)  \\
=& \int_{\Gamma _{0}}\sum_{\xi \subset \eta }\left( 1-\delta \right)
^{\left\vert \xi \right\vert }\int_{\Gamma _{0}}\left( z\delta \right)
^{\left\vert \omega \right\vert }G\left( \xi \cup \omega \right)
\prod\limits_{y\in \xi }e^{-E^{\phi }\left( y,\omega \right) } \\
& \times \prod\limits_{y\in \eta \setminus \xi }\left( e^{-E^{\phi }\left(
y,\omega \right) }-1\right) d\lambda \left( \omega \right) k\left( \eta
\right) d\lambda \left( \eta \right)  \\
=& \int_{\Gamma _{0}}\int_{\Gamma _{0}}\left( 1-\delta \right) ^{\left\vert
\xi \right\vert }\int_{\Gamma _{0}}\left( z\delta \right) ^{\left\vert
\omega \right\vert }G\left( \xi \cup \omega \right) \prod\limits_{y\in \xi
}e^{-E^{\phi }\left( y,\omega \right) } \\
& \times \prod\limits_{y\in \eta }\left( e^{-E^{\phi }\left( y,\omega
\right) }-1\right) d\lambda \left( \omega \right) k\left( \eta \cup \xi
\right) d\lambda \left( \xi \right) d\lambda \left( \eta \right)  \\
=& \int_{\Gamma _{0}}\int_{\Gamma _{0}}\sum_{\omega \subset \xi }\left(
1-\delta \right) ^{\left\vert \xi \setminus \omega \right\vert }\left(
z\delta \right) ^{\left\vert \omega \right\vert }G\left( \xi \right)
\prod\limits_{y\in \xi \setminus \omega }e^{-E^{\phi }\left( y,\omega
\right) } \\
& \times \prod\limits_{y\in \eta }\left( e^{-E^{\phi }\left( y,\omega
\right) }-1\right) k\left( \eta \cup \xi \setminus \omega \right) d\lambda
\left( \xi \right) d\lambda \left( \eta \right) ,
\end{align*}%
therefore,%
\begin{align}\label{apprfordual}
(\hat{P}_{\delta }^{\ast }k)\left( \eta \right)  =&\sum_{\omega \subset \eta
}\left( 1-\delta \right) ^{\left\vert \eta \setminus \omega \right\vert
}\left( z\delta \right) ^{\left\vert \omega \right\vert }\prod\limits_{y\in
\eta \setminus \omega }e^{-E^{\phi }\left( y,\omega \right) } \\
& \times \int_{\Gamma _{0}}\prod\limits_{y\in \xi }\left( e^{-E^{\phi
}\left( y,\omega \right) }-1\right) k\left( \xi \cup \eta \setminus \omega
\right) d\lambda \left( \xi \right) .\notag
\end{align}

\begin{proposition}\label{imp_prop}
Suppose that \eqref{verysmallparam} and \eqref{new_z} are fulfilled. Then, for any $k\in D_{\a C}$ and $\alpha\in(\alpha_{0},\,1)$, where $\alpha_0$ is chosen as in the proof of Proposition~\ref{invariantspace},
\begin{equation}  \label{apprrenest}
\lim_{\delta\rightarrow 0}\biggl\Vert \frac{1}{\delta}( {\hat{P}}
^\ast_{\delta}-1\!\!1 )k - {\hat{L}}^{\odot\a} k\biggr\Vert_{{\mathcal{K}}_C}
=0.
\end{equation}
\end{proposition}

\begin{proof}
Let us recall \eqref{dual-descent} and define
\begin{align*}
(\hP_{\delta }^{\ast, (0) }k)\left( \eta \right) =&\,(1-\delta)^\n k(\eta);\\
(\hP_{\delta }^{\ast, (1) }k)\left( \eta \right) =&\,z\delta
\sum_{x\in \eta
}\left( 1-\delta \right)^{\left\vert \eta \right\vert -1} e_\la\left( e^{-\phi \left( x -\cdot\right) },\eta
\setminus x \right) \notag\\&\times\int_{\Gamma_{0}}e_{\lambda
}\left( e^{-\phi \left( x -\cdot\right)
}-1,\xi \right) k\left( \xi \cup \eta \setminus x
\right) d\lambda \left( \xi \right);
\end{align*}
and $\hP_{\delta }^{\ast, (\geq 2) } = \hP_{\delta }^{\ast } -
\hP_{\delta }^{\ast, (0) } -  \hP_{\delta }^{\ast, (1)
}$.

We will use the following elementary inequality, for any
$n\in\N\cup\{0\}$, $\delta\in(0;1)$
\begin{align*}
0 \leq n-   \frac{1-(1-\delta)^n}{\delta}\leq\delta \frac{n(n-1)}{2}.
\end{align*}
Then, for any $k\in\K_\aC$ and $\lambda$-a.a. $\eta\in\Gamma_{0}$, $\eta\neq\emptyset$
\begin{align}
&C^{-\n}\biggl\vert\frac{1}{\delta}( \hP^{\ast,(0)}_{\delta,\eps}-\1 )k(\eta) + |\eta|k(\eta)\biggr\vert\notag\\
\leq& \,\Vert k\Vert_{\K_\aC} \a^\n \biggl\vert \n-
\frac{1-(1-\delta)^\n}{\delta}\biggr\vert\leq \frac{\delta}{2} \Vert
k\Vert_{\K_\aC} \a^\n \n (\n-1)\label{eq1}
\end{align}
and the function $\a^x x(x-1)$ is bounded for $x\geq 1$, $\a\in(0;1)$.
Next, for any $k\in\K_\aC$ and $\lambda$-a.a. $\eta\in\Gamma_{0}$, $\eta\neq\emptyset$
\begin{align}
&C^{-\n}\biggl\vert\frac{1}{\delta} \hP ^{\ast,(1)}_{\delta}k(\eta)
-z\sum_{x\in\eta}\int_{\Ga_0} e_{\la  }\left( e^{- \phi \left(
x-\cdot \right) },\eta\setminus x \right)
\notag\\&\qquad\qquad\times e_{\la }\left( e^{-\phi \left( x-\cdot
\right) }-1,\xi   \right)
k\left( \xi \cup \eta\setminus x\right) d\la(\xi)\biggr\vert\notag\\
\leq\, & \Vert k\Vert_{\K_\aC}
\frac{z}{\aC}\a^\n\sum_{x\in\eta}\bigl(1-\left( 1-\delta
\right)^{\left\vert \eta \right\vert -1} \bigr)
\int_{\Gamma_{0}}e_{\lambda }\left( \a C \bigl(e^{-\phi \left( x-\cdot \right) }-1\bigr),\xi \right)  d\lambda \left( \xi \right)\notag\\
 \leq\, & \Vert k\Vert_{\K_\aC} \frac{z}{\aC} \a^\n\sum_{x\in\eta}\bigl(1-\left( 1-\delta \right)^{\left\vert \eta
\right\vert -1} \bigr) \exp{\{\a C C_\phi\}}\notag\\
\leq\,& \Vert k\Vert_{\K_\aC}\frac{z}{\aC} \a^\n \delta \n (\n-1)
\exp{\{\a C C_\phi\}}.\label{eq2}
\end{align}
which is small in $\delta$ uniformly by $\n$. Now, using inequality
\[
1-e^{-E^\phi\left( y,\omega \right) }=1-\prod\limits_{x\in \omega
}e^{-\phi \left( x-y\right) }\leq \sum_{x\in \omega }\left(
1-e^{-\phi \left( x-y\right) }\right),
\]
we obtain
\begin{align}
&\frac{1}{\delta }C^{-\left\vert \eta \right\vert } \sum_{\substack{
\omega \subset \eta  \\ \left\vert \omega \right\vert \geq 2}}
\left( 1-\delta \right)^{\left\vert \eta \setminus \omega
\right\vert }\left( z\delta \right)^{\left\vert \omega \right\vert
}e_\la\left( e^{- E^{\phi }\left( \cdot ,\omega \right) },\eta
\setminus \omega \right)\notag\\&\qquad\times\int_{\Gamma_{0} }e_\la\left(
\Bigl\vert e^{- E^{\phi }\left( \cdot ,\omega \right) }-1
\Bigr\vert,\xi \right) |k( \xi
\cup \eta \setminus \omega ) | d\lambda \left( \xi \right) \notag \\
=\,&\Vert k\Vert_{\K_\aC}\alpha^{\left\vert \eta \right\vert }\frac{1}{\delta }\sum_{\substack{ %
\omega \subset \eta  \\ \left\vert \omega \right\vert \geq 2}}\left(
1-\delta \right)^{\left\vert \eta \setminus \omega \right\vert
}\left( \frac{z\delta }{\alpha C}\exp \left\{ \alpha C
C_\phi\right\} \right)
^{\left\vert \omega \right\vert }; \notag\\
\intertext{recall that $\a>\a_0$, therefore, $z\exp\{\aC
C_\phi\}\leq \aC$, and one may continue} \leq\, &\Vert k\Vert_{\K_\aC}
\alpha ^{\left\vert \eta \right\vert }\frac{1}{\delta }\sum
_{\substack{ \omega \subset \eta  \\ \left\vert \omega \right\vert \geq 2}}%
\left( 1-\delta \right)^{\left\vert \eta \setminus \omega
\right\vert }\delta^{\left\vert \omega \right\vert }\notag\\=\,&\Vert
k\Vert_{\K_\aC}\delta \alpha ^{\left\vert \eta \right\vert
}\sum_{k=2}^{\left\vert \eta \right\vert }\frac{\left\vert \eta
\right\vert !}{k!\left( \left\vert \eta \right\vert -k\right)
!}\left(
1-\delta \right)^{\left\vert \eta \right\vert -k}\delta^{k-2} \notag\\
=\, &\Vert k\Vert_{\K_\aC}\delta \alpha^{\left\vert \eta \right\vert
}\sum_{k=0}^{\left\vert \eta \right\vert -2}\frac{\left\vert \eta
\right\vert !}{\left( k+2\right) !\left( \left\vert \eta \right\vert
-k-2\right) !}\left( 1-\delta \right)
^{\left\vert \eta \right\vert -k-2}\delta^{k} \notag\\
=\, &\Vert k\Vert_{\K_\aC}\delta \alpha^{\left\vert \eta \right\vert
}\left\vert \eta \right\vert \left( \left\vert \eta \right\vert
-1\right) \sum_{k=0}^{\left\vert \eta \right\vert -2}\frac{\left(
\left\vert \eta \right\vert -2\right) !}{\left( k+2\right) !\left(
\left\vert \eta \right\vert -k-2\right) !}\left( 1-\delta
\right)^{\left\vert \eta \right\vert -2-k}\delta^{k} \notag\\
\leq\,  &\Vert k\Vert_{\K_\aC}\delta \alpha^{\left\vert \eta
\right\vert }\left\vert \eta \right\vert \left( \left\vert \eta
\right\vert -1\right) \sum_{k=0}^{\left\vert \eta \right\vert
-2}\frac{\left( \left\vert \eta
\right\vert -2\right) !}{k!\left( \left\vert \eta \right\vert -k-2\right) !}%
\left( 1-\delta \right)^{\left\vert \eta \right\vert -2-k}\delta
^{k}\notag\\=\,  &\Vert k\Vert_{\K_\aC}\delta \alpha^{\left\vert \eta
\right\vert }\left\vert \eta \right\vert \left( \left\vert \eta
\right\vert -1\right).\label{eq3}
\end{align}
Combining inequalities \eqref{eq1}--\eqref{eq3} we obtain
\eqref{apprrenest}.
\end{proof}

We recall now well-known approximation result (cf., e.g.,
{\cite[Theorem 6.5]{EK}})
\begin{lemma}\label{EK-lemma}
Let $L, L_n$, $n\in\N$ be Banach spaces, and $p_n: L\rightarrow L_n$
be bounded linear transformation, such that $\sup_n \|p_n\|<\infty
$. For any $n\in\N$, let $T_n$ be a linear contraction on $L_n$, let
$\eps_n>0$ be such that $\lim_{n\rightarrow \infty} \eps_n =0$, and
put $A_n=\eps_n^{-1}(T_n - \1)$. Let $T_t$ be a strongly continuous
contraction semigroup on $L$ with generator $A$ and let $D$ be a
core for $A$. Then the following are equivalent:
\begin{enumerate}
\item For each $f\in L$, $||T_n^{[t/\eps_n]} p_n f- p_n
T_t f||_{L_n}\rightarrow 0$, $n\rightarrow\infty$ for all $t\geq0$ uniformly on bounded intervals.
Here and below $[\,\cdot\,\,]$ mean the entire part of a real
number.

\item For each $f\in D$, there exists $f_n\in L_n$ for each
$n\in\N$ such that

$||f_n-p_n f||_{L_{n}}\rightarrow 0$ and $||A_n f_n -
p_n Af||_{L_n}\rightarrow0$, $n\rightarrow\infty$.
\end{enumerate}
\end{lemma}

\begin{theorem}
Let $\alpha_0$ be chosen as in the proof of the Proposition~\ref%
{invariantspace} and be fixed. Let $\alpha\in(\alpha_0;1)$ and
$k\in\overline{{\mathcal{K} }_{\a C}}$ be given. Then
\begin{equation*}
({\hat{P}} ^\ast_{\delta})^{[{t}/{\delta}]}k\rightarrow{\hat{T}%
} ^{\odot\a}(t)k, \quad \delta\rightarrow 0
\end{equation*}
in the space $\overline{{\mathcal{K}}_{\a C}}$ with norm
$\|\cdot\|_{{\mathcal{K}}_C}$ for all $t\geq 0$ uniformly on bounded
intervals.
\end{theorem}
\begin{proof}
We may apply Proposition~\ref{imp_prop} to use
Lemma~\ref{EK-lemma} in the case $L_n=L=\overline{\L_\aC}$,
$p_n=\1$, $f_n=f=k$, $\eps_n=\delta\rightarrow0$, $n\in\N$.
\end{proof}

\section{Positive definiteness}

\begin{definition}
A measurable function $k:\Gamma_0\rightarrow{\mathbb{R}}$ is called a
positive defined function (cf. \cite{Le75I, Le75II}) if for any $G\in
L_{\mathrm{ls}}^0(\Ga _0)$ such that $KG\geq 0$ and $G\in\mathcal{L}_{C}$ for some $C>1$ the following
inequality holds%
\begin{equation*}
\int_{\Gamma _{0}}G\left( \eta \right) k\left( \eta \right) d\lambda \left(
\eta \right) \geq 0.
\end{equation*}
\end{definition}
\noindent In \cite{Le75I, Le75II}, it was shown that if $k$ is a positive defined
function and $|k(\eta)|
\leq C^{|\eta|} (|\eta|!)^2$, $\eta\in\Ga_0$ then there exists a unique measure $\mu\in{%
\mathcal{M}}^1_{\mathrm{fm}}(\Gamma)$ such that $k=k_\mu$ be its
correlation functional in the sense of \eqref{eqmeans}. Our aim is to show that the evolution
$k\mapsto \hat{T}_{t}^{\odot }k$ preserves the property of the positive definiteness.

\begin{theorem}
Let \eqref{verysmallparam} holds and $k\in \overline{D(\hat{L}^{\ast })}%
\subset \mathcal{K}_{C}$ be a positive defined function. Then $k_{t}:=\hat{T}_{t}^{\odot }k\in \overline{D(\hat{L}^{\ast })}%
\subset \mathcal{K}_{C}$ will be a positive defined function for any $t\geq0$.
\end{theorem}

\begin{proof}
Let $C>0$ be arbitrary and fixed.
For any $G\in \mathcal{L}_{C}^{\mathrm{ls}}$ we have
\begin{equation}\label{sdual}
\int_{\Gamma _{0}}G\left( \eta \right) k_{t}\left( \eta \right)
d\lambda \left( \eta \right) =\int_{\Gamma _{0}}(\hat{T}_{t}G)\left(
\eta \right) k\left( \eta \right) d\lambda \left( \eta \right) .
\end{equation}
By \cite[Proposition 3.10]{FKKZh}, under condition
\eqref{verysmallparam}, we obtain that
\[
\lim_{n\rightarrow 0}\int_{\Gamma _{\Lambda _{n}}}\left\vert T_{n}^{\left[ nt%
\right] }\1_{\Gamma _{\Lambda _{n}}}G\left( \eta \right) -\1_{\Gamma
_{\Lambda _{n}}}(\eta)(\hat{T}_{t}G)\left( \eta \right) \right\vert
C^{\left\vert \eta \right\vert }d\lambda \left( \eta \right) =0,
\]%
where for $n\geq 2$%
\[
T_{n}=\hat{P}_{\frac{1}{n}}^{\Lambda _{n}}
\]%
and $\La_n \nearrow\X.$ Note that, by the dominated convergence
theorem,
\begin{align*}
\int_{\Gamma _{0}}(\hat{T}_{t}G)\left( \eta \right) k\left( \eta
\right) d\lambda \left( \eta \right)  =&\lim_{n\rightarrow \infty
}\int_{\Gamma _{0}}\1_{\Gamma _{\Lambda _{n}}}\left( \eta \right)
(\hat{T}_{t}G)\left( \eta
\right) k\left( \eta \right) d\lambda \left( \eta \right)  \\
=&\lim_{n\rightarrow \infty }\int_{\Gamma _{\Lambda _{n}}}(\hat{T}%
_{t}G)\left( \eta \right) k\left( \eta \right) d\lambda \left( \eta
\right) .
\end{align*}%
Next,%
\begin{align*}
&\left\vert \int_{\Gamma _{\Lambda _{n}}}(\hat{T}_{t}G)\left( \eta
\right) k\left( \eta \right) d\lambda \left( \eta \right)
-\int_{\Gamma _{\Lambda _{n}}}T_{n}^{\left[ nt\right] }\1_{\Gamma
_{\Lambda _{n}}}G\left( \eta
\right) k\left( \eta \right) d\lambda \left( \eta \right) \right\vert  \\
\leq &\int_{\Gamma _{\Lambda _{n}}}\left\vert T_{n}^{\left[
nt\right]
}\1_{\Gamma _{\Lambda _{n}}}G\left( \eta \right) -\1_{\Gamma _{\Lambda _{n}}}(\eta)(%
\hat{T}_{t}G)\left( \eta \right) \right\vert k\left( \eta \right)
d\lambda
\left( \eta \right)  \\
\leq  & \, \|k\|_{\K_C}\int_{\Gamma _{\Lambda _{n}}}\left\vert T_{n}^{\left[
nt\right]
}\1_{\Gamma _{\Lambda _{n}}}G\left( \eta \right) -\1_{\Gamma _{\Lambda _{n}}}(\eta)(%
\hat{T}_{t}G)\left( \eta \right) \right\vert C^{\left\vert \eta
\right\vert }d\lambda \left( \eta \right) \rightarrow
0,~~n\rightarrow \infty .
\end{align*}%
Therefore,%
\begin{equation}\label{eq-dop}
\int_{\Gamma _{0}}(\hat{T}_{t}G)\left( \eta \right) k\left( \eta
\right) d\lambda \left( \eta \right) =\lim_{n\rightarrow \infty
}\int_{\Gamma _{\Lambda _{n}}}T_{n}^{\left[ nt\right] }\1_{\Gamma
_{\Lambda _{n}}}G\left( \eta \right) k\left( \eta \right) d\lambda
\left( \eta \right) .
\end{equation}
Our aim is to show that for any $G\in \mathcal{L}_{C}^{\mathrm{ls}}$ the inequality $KG\geq0$ implies
\[
\int_{\Gamma _{0}}G\left( \eta \right) k_{t}\left( \eta \right)
d\lambda \left( \eta \right) \geq 0.
\]
By \eqref{sdual} and \eqref{eq-dop}, it is enough to show that for any $m\in
\mathbb{N}$ and for any $G\in \mathcal{L}_{C}^{\mathrm{ls}}$ such that $KG\geq0$
the following inequality holds
\begin{equation}\label{eq-dop2}
\int_{\Gamma _{0}}\1_{\Gamma _{\Lambda _{n}}}T_{n}^{m}\1_{\Gamma
_{\Lambda _{n}}}G\left( \eta \right) k\left( \eta \right) d\lambda
\left( \eta \right) \geq 0, \quad m\in\N_0.
\end{equation}
The inequality \eqref{eq-dop2} is fulfilled  if only
\begin{equation}\label{dop12}
K\1_{\Gamma _{\Lambda _{n}}}T_{n}^{m}G_{n}\geq 0,
\end{equation}
where $G_{n}:=\1_{\Gamma _{\Lambda _{n}}}G$.
Note that
\begin{align}
\bigl( K\1_{\Gamma _{\Lambda _{n}}}T_{n}^{m}G_{n}\bigr) \left(
\gamma \right)  =&\sum_{\eta \Subset \gamma }\1_{\Gamma _{\Lambda
_{n}}}\left( \eta
\right) \left( T_{n}^{m}G_{n}\right) \left( \eta \right)  \label{dop23}\\
=&\sum_{\eta \subset \gamma _{\Lambda _{n}}}\left(
T_{n}^{m}G_{n}\right) \left( \eta \right) =\left(
KT_{n}^{m}G_{n}\right) \left( \gamma _{\Lambda _{n}}\right) \notag
\end{align}
for any $m\in\N_0$. In particular,
\begin{equation}\label{dop123}
\left( KG_{n}\right) \left( \gamma \right) =\left( K\1_{\Gamma
_{\Lambda _{n}}}G\right) \left( \gamma \right) =\left( KG\right)
\left( \gamma _{\Lambda _{n}}\right) \geq 0.
\end{equation}

Let us now consider any $\tilde{G}\in \mathcal{L}_{C}^{\mathrm{ls}}$ (stress that $\tilde{G}$ is
not necessary equal to $0$ outside of $\Ga_{\Lambda _{n}}$) and
suppose that $\bigl( K \tilde{G}\bigr) \left( \gamma \right) \geq 0$
for any $\gamma \in \Gamma _{\Lambda _{n}}$. Then
\begin{align}\label{dop345}
&\bigl( KT_{n}\tilde{G}\bigr) \left( \gamma _{\Lambda _{n}}\right) = \bigl( K\hat{P}_{\frac{1}{n}}^{\Lambda _{n}}\tilde{G}\bigr)
\left(
\gamma _{\Lambda _{n}}\right) =\bigl( P_{\frac{1}{n}}^{\Lambda _{n}}K\tilde{G}\bigr) \left(
\gamma
_{\Lambda _{n}}\right)  \\
=& \Bigl( \Xi _{\frac{1}{n}}^{\Lambda _{n}}\left( \gamma _{\Lambda
_{n}}\right) \Bigr) ^{-1}\sum_{\eta \subset \gamma _{\Lambda
_{n}}}\biggl(
\frac{1}{n}\biggr) ^{\left\vert \eta \right\vert }\biggl( 1-\frac{1}{n}%
\biggr) ^{\left\vert \gamma \setminus \eta \right\vert } \notag\\
&\times \int_{\Gamma _{\Lambda _{n}}}\biggl( \frac{z}{n}\biggr)
^{\left\vert \omega \right\vert }\prod_{y\in \omega }e^{-E^{\phi
}\left( y,\gamma \right) }\bigl( K\tilde{G}\bigr) \bigl( \left(
\gamma _{\Lambda _{n}}\setminus \eta
\right) \cup \omega \bigr) d\lambda \left( \omega \right)
\geq 0.\notag
\end{align}
By \eqref{dop123}, setting $\tilde{G}=G_n\in\mathcal{L}_{C}^{\mathrm{ls}}$ we obtain, because of \eqref{dop345}, $KT_nG_n \geq0$.
Next, setting $\tilde{G}=T_nG_n\in\mathcal{L}_{C}^{\mathrm{ls}}$ we obtain, by \eqref{dop345}, $KT_n^2 G_n \geq0$.
Then, using an induction mechanism, we obtain that%
\[
\left( KT_{n}^{m}G_{n}\right) \left( \gamma _{\Lambda _{n}}\right)
\geq 0, \quad m\in\N_0,
\]
that, by \eqref{dop12} and \eqref{dop23}, yields \eqref{eq-dop2}.
This completes the proof. \end{proof}

\section{Ergodicity}
Let $k\in\overline{{\mathcal{K}}_{\a C}}$ be such that $k(\emptyset)=0$
then, by \eqref{apprfordual}, $(\hat{P}_{\delta }^{\ast }k)\left( \emptyset \right)=0$. Class of all such functions we denote by
$\K_\a^0$.
\begin{proposition}\label{propergod}
Assume that there exists $\nu\in(0;1)$ such that
\begin{equation}  \label{nu-verysmallparam}
z\leq \min\Bigl\{\nu Ce^{-CC_{\phi }} ; \, 2Ce^{-2CC_{\phi }}\Bigr\}.
\end{equation}
Let, additionally, $\a\in(\a_0;1)$,
where $\alpha_0$ is chosen as in the proof of the Proposition~\ref{invariantspace}.
Then for any $\delta \in(0;1)$ the following estimate holds
\begin{equation}\label{supercontraction}
\Bigl\| \hat{P}_{\delta }^{\ast }\! \upharpoonright_{\K_\a^{0}} \Bigr\|\leq 1-(1-\nu)\delta.
\end{equation}
\end{proposition}
\begin{proof}
It is easily seen that for any $k\in\K_\a^0$ the following
inequality holds
\[
\left\vert k\left( \eta \right) \right\vert \leq \\1_{\left\vert
\eta \right\vert >0}\left\Vert k\right\Vert
_{\mathcal{K}_{C}}C^{\left\vert \eta \right\vert },
\quad\lambda\mathrm{-a.a.}\;\;\eta\in\Ga_0.
\]%
Then, using \eqref{apprfordual}, we have
\begin{align*}
&C^{-\left\vert \eta \right\vert }\left\vert (\hat{P}_{\delta
}^{\ast
}k)\left( \eta \right) \right\vert  \\
\leq &C^{-\left\vert \eta \right\vert }\sum_{\omega \subset \eta
}\left( 1-\delta \right) ^{\left\vert \eta \setminus \omega
\right\vert }\left( z\delta \right) ^{\left\vert \omega \right\vert
}\int_{\Gamma _{0}}\prod\limits_{y\in \xi }\left( 1-e^{-E^{\phi
}\left( y,\omega \right) }\right) \left\vert k\left( \xi \cup \eta
\setminus \omega \right)
\right\vert d\lambda \left( \xi \right)  \\
\leq &\left\Vert k\right\Vert _{\mathcal{K}_{C}}\sum_{\omega \subset
\eta }\left( 1-\delta \right) ^{\left\vert \eta \setminus \omega
\right\vert }\left( \frac{z\delta }{C}\right) ^{\left\vert \omega
\right\vert }\int_{\Gamma _{0}}\prod\limits_{y\in \xi }\left(
1-e^{-E^{\phi }\left( y,\omega \right) }\right) C^{\left\vert \xi
\right\vert }\1_{\left\vert \xi \right\vert +\left\vert \eta
\setminus \omega \right\vert >0}d\lambda \left(
\xi \right)  \\
=&\left\Vert k\right\Vert _{\mathcal{K}_{C}}\sum_{\omega \subsetneq
\eta }\left( 1-\delta \right) ^{\left\vert \eta \setminus \omega
\right\vert }\left( \frac{z\delta }{C}\right) ^{\left\vert \omega
\right\vert }\int_{\Gamma _{0}}\prod\limits_{y\in \xi }\left(
1-e^{-E^{\phi }\left( y,\omega \right) }\right) C^{\left\vert \xi
\right\vert }d\lambda \left( \xi
\right)  \\
&+\left\Vert k\right\Vert _{\mathcal{K}_{C}}\left( \frac{z\delta }{C}%
\right) ^{\left\vert \eta \right\vert }\int_{\Gamma
_{0}}\prod\limits_{y\in \xi }\left( 1-e^{-E^{\phi }\left( y,\omega
\right) }\right) C^{\left\vert \xi \right\vert }\1_{\left\vert \xi
\right\vert >0}d\lambda \left( \xi
\right)  \\
=&\left\Vert k\right\Vert _{\mathcal{K}_{C}}\sum_{\omega \subsetneq
\eta }\left( 1-\delta \right) ^{\left\vert \eta \setminus \omega
\right\vert }\left( \frac{z\delta }{C}\right) ^{\left\vert \omega
\right\vert }\int_{\Gamma _{0}}\prod\limits_{y\in \xi }\left(
1-e^{-E^{\phi }\left( y,\omega \right) }\right) C^{\left\vert \xi
\right\vert }d\lambda \left( \xi
\right)  \\
&+\left\Vert k\right\Vert _{\mathcal{K}_{C}}\left( \frac{z\delta }{C}%
\right) ^{\left\vert \eta \right\vert }\int_{\Gamma
_{0}}\prod\limits_{y\in \xi }\left( 1-e^{-E^{\phi }\left( y,\omega
\right) }\right) C^{\left\vert
\xi \right\vert }d\lambda \left( \xi \right) -\left\Vert k\right\Vert _{%
\mathcal{K}_{C}}\left( \frac{z\delta }{C}\right) ^{\left\vert \eta
\right\vert } \\
=&\left\Vert k\right\Vert _{\mathcal{K}_{C}}\sum_{\omega \subset
\eta }\left( 1-\delta \right) ^{\left\vert \eta \setminus \omega
\right\vert }\left( \frac{z\delta }{C}\right) ^{\left\vert \omega
\right\vert }\int_{\Gamma _{0}}\prod\limits_{y\in \xi }\left(
1-e^{-E^{\phi }\left( y,\omega \right) }\right) C^{\left\vert \xi
\right\vert }d\lambda \left( \xi
\right) \\&-\left\Vert k\right\Vert _{\mathcal{K}_{C}}\left( \frac{z\delta }{C}%
\right) ^{\left\vert \eta \right\vert } \\
=&\left\Vert k\right\Vert _{\mathcal{K}_{C}}\sum_{\omega \subset
\eta }\left( 1-\delta \right) ^{\left\vert \eta \setminus \omega
\right\vert }\left( \frac{z\delta }{C}\right) ^{\left\vert \omega
\right\vert }\exp \left\{ C\int_{\mathbb{R}^{d}}\left( 1-e^{-E^{\phi
}\left( y,\omega \right)
}\right) dy\right\} \\&-\left\Vert k\right\Vert _{\mathcal{K}_{C}}\left( \frac{%
z\delta }{C}\right) ^{\left\vert \eta \right\vert } \\
\leq &\left\Vert k\right\Vert _{\mathcal{K}_{C}}\sum_{\omega
\subset \eta }\left( 1-\delta \right) ^{\left\vert \eta \setminus
\omega \right\vert }\left( \frac{z\delta }{C}\right) ^{\left\vert
\omega \right\vert }\exp \left\{ CC_{\beta }\left\vert \omega
\right\vert \right\} -\left\Vert k\right\Vert
_{\mathcal{K}_{C}}\left( \frac{z\delta }{C}\right) ^{\left\vert
\eta \right\vert } \\
\leq &\left\Vert k\right\Vert _{\mathcal{K}_{C}}\sum_{\omega
\subset \eta }\left( 1-\delta \right) ^{\left\vert \eta \setminus
\omega \right\vert }\left( \nu  \delta \right) ^{\left\vert \omega
\right\vert }-\left\Vert k\right\Vert _{\mathcal{K}_{C}}\left(
\frac{z\delta }{C}\right) ^{\left\vert
\eta \right\vert } \\
=&\left\Vert k\right\Vert _{\mathcal{K}_{C}}\left( \left( 1-\left(
1-\nu  \right) \delta \right) ^{\left\vert \eta \right\vert }-\left(
\frac{z\delta
}{C}\right) ^{\left\vert \eta \right\vert }\right)  \\
=&\left\Vert k\right\Vert _{\mathcal{K}_{C}}\left( 1-\left( 1-\nu
\right) \delta -\frac{z\delta }{C}\right) \sum_{j=0}^{\left\vert
\eta \right\vert -1}\left( 1-\left( 1-\nu  \right) \delta \right)
^{\left\vert
\eta \right\vert -1-\left\vert j\right\vert }\left( \frac{z\delta }{C}%
\right) ^{j} \\
\leq &\left\Vert k\right\Vert _{\mathcal{K}_{C}}\left( 1-\left(
1-\nu  \right) \delta -\frac{z\delta }{C}\right)
\sum_{j=0}^{\left\vert \eta
\right\vert -1}\left( \frac{z\delta }{C}\right) ^{j} \\
=&\left\Vert k\right\Vert _{\mathcal{K}_{C}}\left( 1-\left( 1-\nu
\right) \delta -\frac{z\delta }{C}\right) \frac{1-\left( \frac{z\delta }{C}%
\right) ^{\left\vert \eta \right\vert }}{1-\frac{z\delta }{C}} \\
\leq &\left\Vert k\right\Vert _{\mathcal{K}_{C}}\left( 1-\left(
1-\nu
\right) \delta -\frac{z\delta }{C}\right) \frac{1}{1-\frac{z\delta }{C}} \\
= &\left\Vert k\right\Vert _{\mathcal{K}_{C}}\left(
1-\frac{\left( 1-\nu  \right) \delta }{1-\frac{z\delta }{C}}\right)
\leq\left\Vert k\right\Vert _{\mathcal{K}_{C}}\bigl( 1-\left( 1-\nu
\right) \delta \bigr),
\end{align*}
where we have used that, clearly, $z<\nu C<C$. The statement is
proved.
\end{proof}

\begin{remark}
Condition \eqref{nu-verysmallparam} is equivalent to \eqref{verysmallparam} and \eqref{new_z}.
\end{remark}

Suppose that (cf. \eqref{less_e-1})
\begin{equation}\label{LAHT}
z C_\phi < (2e)^{-1},
\end{equation}
then (see, e.g., \cite{FiKoLy} for details) there exists a Gibbs measure $\mu$ on $\bigl(\Ga,
\B(\Ga)\bigr)$ corresponding to the potential $\phi\geq 0$ and
activity parameter $z$. We denote the corresponding correlation
function by $k_\mu$. The measure $\mu$ is reversible (symmetrizing) for the operator defined by \eqref{genGa} (see, e.g., \cite{FiKoLy}, \cite{KoLy}). Therefore, for any $F\in K\Bbs(\Ga_0)$
\begin{equation}\label{invGibbs}
\int_\Ga LF(\ga)d\mu(\ga)=0.
\end{equation}

\begin{theorem}\label{thmergod}
Let \eqref{LAHT} and \eqref{nu-verysmallparam} hold and let $\a\in(\a_0;1)$,
where $\alpha_0$ is chosen as in the proof of Proposition~\ref{invariantspace}. Let
$k_0\in\Ka$, $k_t={\hat{T}} ^{\odot\a}(t)k_0$. Then for any $t\geq0$
\begin{equation}\label{ergodineq}
\|k_t -k_\mu\|_{\K_C}\leq e^{-(1-\nu)t} \|k_0 -k_\mu\|_{\K_C}.
\end{equation}
\end{theorem}

\begin{proof}
First of all, let us note that for any $\a\in(\a_0;1)$ the
inequality \eqref{ineq-alpha} implies $z\leq\a C \exp\{-\aC
C_\phi\}$. Hence $k_\mu(\eta)\leq (\aC)^{|\eta|}$, $\eta\in\Ga_0$.
Therefore, $k_\mu\in\K_\aC\subset\Ka\cap D(\hat{L}^\ast)$. By
\eqref{invGibbs}, for any $G\in \Bbs(\Ga_0)$ we have
$\langle\!\langle\hat{L}G, k_\mu\rangle\!\rangle=0$. It means that
$\hat{L}^\ast k_\mu = 0$. Therefore, ${\hat{L}}^{\odot\a}k_\mu=0$.
As a result,
 ${\hat{T}} ^{\odot\a}(t)k_\mu=k_\mu$.
Let $r_0=k_0-k_\mu\in\Ka$. Then $r_0\in\K_a^0$ and
\begin{align*}
&\|k_t -k_\mu\|_{\K_C}= \bigl\| {\hat{T}} ^{\odot\a}(t) r_0\bigr\|_{\K_C}
\\&\leq \Bigl\| \bigl(\hat{P}_{\delta
}^{\ast
}\bigr)^{\left[\frac{t}{\delta}\right]} r_0\Bigr\|_{\K_C}+\Bigl\| {\hat{T}} ^{\odot\a}(t) r_0- \bigl(\hat{P}_{\delta
}^{\ast
}\bigr)^{\left[\frac{t}{\delta}\right]} r_0\Bigr\|_{\K_C}\\
&\leq \Bigl\| \hat{P}_{\delta
}^{\ast
}\upharpoonright_{\K_\a^{0}}\Bigr\|^{\left[\frac{t}{\delta}\right]} \cdot \|r_0\|_{\K_C}+\Bigl\| {\hat{T}} ^{\odot\a}(t) r_0- \bigl(\hat{P}_{\delta
}^{\ast
}\bigr)^{\left[\frac{t}{\delta}\right]} r_0\Bigr\|_{\K_C}\\
&\leq \bigl( 1 -(1-\nu)\delta \bigr)^{\frac{t}{\delta}-1} \|r_0\|_{\K_C}+\Bigl\| {\hat{T}} ^{\odot\a}(t) r_0- \bigl(\hat{P}_{\delta
}^{\ast
}\bigr)^{\left[\frac{t}{\delta}\right]} r_0\Bigr\|_{\K_C},
\end{align*}
since $0<1-(1-\nu)\delta<1$ and
$\frac{t}{\delta}<\bigl[\frac{t}{\delta}\bigr]+1$. Taking the limit
as $\delta\downarrow 0$ in the right hand side of this inequality we
obtain \eqref{ergodineq}.
\end{proof}

\begin{acknowledgements}
The financial support of DFG through the SFB 701 (Bielefeld
University) and German-Ukrainian Project 436 UKR 113/80 and 436 UKR
113/94 is gratefully acknowledged. This work was partially supported by the Marie Curie "Transfer
of Knowledge" program, project TODEQ (Warsaw, IMPAN). O.K. is very
thankful to Prof. J. Zemanek for fruitful and stimulating discussions.
\end{acknowledgements}

\end{document}